\documentclass[a4paper,11pt]{amsart}
\usepackage[utf8]{inputenc}
\usepackage{amsaddr}
\usepackage{amsmath}
\usepackage{calrsfs}
\usepackage{graphicx,xcolor}
\usepackage[colorlinks = true,linkcolor = blue,urlcolor  = blue,citecolor = blue,anchorcolor = blue]{hyperref}
\usepackage{dsfont}
\usepackage{lipsum}
\usepackage{tikz}
\usepackage{forest}
\usepackage{float}
\usepackage{cite}

\numberwithin{equation}{section}

\theoremstyle{definition}
\newtheorem{dfn}{Definition}[section]

\theoremstyle{plain}
\newtheorem{thm}{Theorem}[section]

\newtheorem{cor}{Corollary}[section]
\newtheorem{lmm}{Lemma}[section]
\theoremstyle{definition}
\newtheorem{rem}{Remark}[section]

\newcommand{\N}{\mathbb{N}}
\newcommand{\R}{\mathbb{R}}
\newcommand{\E}{\mathbb{E}}
\renewcommand{\P}{\mathbb{P}}

\newcommand{\F}{\mathcal{F}}

\newcommand{\e}{\mathrm{e}}

\newcommand{\deq}{\overset{d}{=}}

\newcommand{\1}{\mathds{1}}
\renewcommand{\d}{\mathrm{d}}

\newcommand{\pa}{\partial}
\newcommand{\erf}{\mathrm{erf}}
\newcommand{\BS}{{\mathrm{\scriptstyle{BS}}}}
\newcommand{\Call}{{\mathrm{\scriptstyle{Call}}}}

\begin{document}
\title[Hedging in jump diffusion model]
{Hedging in Jump Diffusion Model with Transaction Costs}

\date{\today}

\author[Almani]{\vspace*{-0.5cm}Hamidreza Maleki Almani}
\email{hmaleki@uwasa.fi}

\author[Shokrollahi]{\vspace*{-1cm}Foad Shokrollahi}
\email{foad.shokrollahi@uwasa.fi}

\author[Sottinen]{\vspace*{-1cm}Tommi Sottinen}
\email{tommi.sottinen@iki.fi}

\address{Department of Mathematics and Statistics, School of Technology and Innovation, University of Vaasa, P.O. Box 700, FIN-65101 Vaasa, FINLAND}
\begin{abstract}
We consider the jump-diffusion risky asset model and study its conditional prediction laws. Next, we explain the conditional least square hedging strategy and calculate its closed form for the jump-diffusion model, considering the Black-Scholes framework with interpretations related to investor priorities and transaction costs. We investigate the explicit form of this result for the particular case of the European call option under transaction costs and formulate recursive hedging strategies. Finally, we present a decision tree, table of values, and figures to support our results.
\end{abstract}



\maketitle

\section{Introduction}
The decade of 1970-80 was undoubtedly the initiation of modern finance. During that period, the main issue economists were eager to solve was to evaluate a closed form for the option price. Pioneer among them, Black–Scholes (BS) (1973) \cite{black1973pricing} applied the geometric Brownian motion (gBm) of Samuelson (1965) \cite{samuelson1965rational} adjusted by the theoretical framework of finance, the {\it no-arbitrage} theory. This continuous hedging-pricing model was later employed successfully for pricing of rational options by Merton (1976) \cite{merton1973theory}, and extened for transaction costs and taxes by Ingersoll (1976) \cite{ingersoll1976theoretical}.

The gBm is the stochastic differential equation (SDE):

\begin{equation}\label{eq:gBm}
	\frac{\d S_t}{S_t}\, =\,\mu\,\d t\,+\,\sigma\,\d W_t,
\end{equation}

where $S$ is the asset price, $W$ is the Wiener process, and $\mu,\sigma$ are constants.
In a {\it self-financing} trading strategy, the investor holds a portfolio $\eta,\pi$ of some risk-free and risky assets respectively with prices $A, S$, specially the case that $A$ is the value of the cash account. So, the value $V$ of the associated portfolio at time $t\ge 0$ changes by
\begin{equation}\label{eq:SF_dynamic}
	\d V_t = \eta_t\,\d A_t +\pi_t\,\d S_t.
\end{equation}
Applying It\^o's (1951) \cite{ito1944integral,ito1951sde} calculus, Black--Scholes showed that if
\begin{enumerate}
	\item[(i)] The investment strategy is self-financing,
	\item[(ii)] The interest rate $r$ of the risk-free asset $A$ is constant,
	\item[(iii)] The risky asset $S$ has the gBm dynamic \eqref{eq:gBm},
	\item[(iv)] For some $g\in\mathcal{C}^{1,2}([0,T]\times \R_+)$ we have $V_t = g(t,S_t)$,
\end{enumerate}
then
\begin{equation}\label{eq:BS}
	\frac{\pa g}{\pa t} + rx\frac{\pa g}{\pa x} + \frac12\sigma^2 x^2\frac{\pa^2 g}{\pa x^2} - rg = 0,
\end{equation}
and from (iv), on $[0,T]$
\begin{equation}\label{eq:BS_pi}
	\pi_t = \frac{\pa g}{\pa x}(t,S_t).
\end{equation}

By the no-arbitrage theory, value of the portfolio $V$ must match to the value of the option of the risky asset that the investor holds. So, (iv) is equivalent to having the option price depend only on time $t$ and the asset price $S_t$. This is valid for the European {\it vanilla} options, and for such option's boundary conditions, the equation \eqref{eq:BS} has a unique solution. Later on, Leland (1985) \cite{leland1985option} developed this approach under transaction costs.
For the mathematical reconstruction, see  related references \cite{Denis-Kabanov-2010,Kabanov-Safarian-2009}.\\

As explained above, the $BS$ model makes it possible to identify a {\it perfect hedging} strategy. In other word, if in a market we have all (i)-(iv) valid, postulating of the no-arbitrage theory (which is interesting in equilibrium finance) is possible. Although, this was a marvelous inception that is still applicable on {\it business time} (see \cite{geman2001asset,geman2001time}), some criticisms arose later. In fact, the main criticism was even published some years before Black--Scholes work, that is the Mandelbrot \& Van Ness (1968) \cite{mandelbrot1968fractional}. Working on stock price data, they had concluded the noise of price is far different with a Wiener process. This rules out the gBm model \eqref{eq:gBm} and also returns the possibility of arbitrage.\\

In modern finance, to overcome such issues, mathematical economists have applied (or some introduced) other stochastic processes in \eqref{eq:gBm} place of Wiener process. Some considered the {\it stable} processes (see \cite{ito2006essentials,mandelbrot1997variation,samoradnitsky-Taqqu-1994}), some tried processes with {\it long memory} (see \cite{Azmoodeh-2013,Shokrollahi-Kilicman-Magdziarz-2016,Wang-2010a,Wang-2010b,Wang-Yan-Tang-Zhu-2010,Wang-Zhu-Tang-Yan-2010}), and most recently they have studied the processes including {\it jumps} (see \cite{carr2002fine,geman2002pure,madan2001financial,kou2002jump,kou2001option,MERTON1976125}). The simplest model with these jump processes \eqref{eq:gBm} is the {\it jump diffusion (JD)} model
\begin{equation}\label{eq:JD}
	\frac{\d S_t}{S_{t^-}} = \mu\, \d t + \sigma\, \d W_t + \d J_t.
\end{equation}
For the JD model with compound Poisson jump of independent normal processes, \cite{MERTON1976125} evaluated the hedging strategy and option prices. \cite{das1996exact} also calculated a closed forms of hedging and pricing when the interest rate has the JD model. Although the general JD model nicely calibrates the price data, it allows arbitrage opportunities. So, the perfect hedging is impossible for the general JD model (see \cite{lamberton2011introduction} chapter 7).\\

For models where the BS hedging method is not applicable, economists investigated other hedging methods. F\"ollmer and Schweizer (1991) \cite{follmer1991hedging} initiated hedging associated to the {\it risk measure} and Schweizer (1992-95) developed this idea in \cite{schweizer1992mean,schweizer1994approximating,schweizer1995minimal}. However, as the BS is still valid on {\it business time}, a financial investor naturally may ask the following question:
\begin{center}
	{\it ``In an incomplete market, what is the discrete strategy with the \textbf{closest price} to the BS strategy of the associated background complete market?"}
\end{center}
Motivated to answer this, Sottinen \& Viitasaari (2018) \cite{sottinen2018conditional} introduced the {\it conditional mean hedging} (CMH) method by employing the conditional laws from their prior paper \cite{Sottinen-Viitasaari-2016-preprintb}. This hedging is based on the coalitional average of the new model's value of portfolio, with respect to the coalitional average of the BS portfolio's value. Shokrollahi \& Sottinen (2017) \cite{SHOKROLLAHI201785} studied this method for the {\it fractional} Black-Scholes model.\\

In our recent article \cite{almani2024prediction}, we studied the required conditional laws for the Volterra-jump models, which JD is a special case of. We aimed to find the CMH strategy when the underlying asset has the general JD model. However, considering the numerical results revealed a critical fact:
\begin{center}
	{\it ``Even if the conditional means of some stochastic processes are equal, still the \textbf{minimum distance} of those processes is not guaranteed!"}
\end{center}
Therefore, we considered a stronger approach. The idea is to investigate a discrete strategy that has the minimum conditional difference with the BS strategy. This is the {\it conditional least-square hedging} (CLH). Surprisingly this stronger hedging strategy is possible for the JD model. Using the conditional formulas of our previous article, we calculate the (CLH) strategy for the JD model under transaction costs here.\\

The rest of this paper is as follows. In Section \ref{sect: pre} we study the conditional prediction laws for the jump diffusion model. In Section \ref{sect: CMH} we consider the CMH strategies for this model under proportional transaction costs with a decision tree for the investor. In Section \ref{sect: CLH} we explain and investigate the CLH strategies for it. In Section \ref{sect: EU_call} we investigate the explicit form of CLH strategies for the European call options. In Section \ref{sect: simulations} we visualize our finding by some simulations and figures.\\

\section{Preliminaries}\label{sect: pre}

We consider the discounted pricing model where the riskless interest rate $r$ is zero ($\d A = 0$) and the risky asset $S$ is given by the following jump--diffusion stochastic differential equation (SDE)
\begin{equation}\label{eq:sde-S}
\frac{\d S_t}{S_{t^-}} = \mu\, \d t + \sigma\, \d W_t + \d J_t,
\end{equation}
where $W$ is the standard Brownian motion and $J$ is an independent compound Poisson jump process with intensity $\lambda$ and jump distribution $F$. In other words
$$
J_t = \sum_{k=1}^{N_t} \xi_k,
$$
where $N=(N_t)_{t\in[0,T]}$ is a Poisson process with intensity $\lambda$ and the jumps $\xi_k$, $k\in\N$, are i.i.d. with common distribution $F$, and they are independent of the Poisson process $N$ and the Brownian motion $W$. We denote
\begin{eqnarray*}
	\epsilon_1 &=& \E[\xi_k], \\
	\epsilon_2 &=& \E[\xi_k^2].
\end{eqnarray*}
The path-wise solution of the SDE \eqref{eq:sde-S} is
\begin{equation}\label{eq:S}
S_t = S_0\,\e^{(\mu - \frac{\sigma^2}{2})t + \sigma W_t}\prod_{k=1}^{N_t}(1+\xi_k),
\end{equation}
see e.g. \cite{lamberton2011introduction}. Note that as the asset is positive $S_t>0, \forall t\ge 0$, we must have $\xi_k>-1, \forall k\ge 1$.\\

Now, as in real world, continuous trading is impossible, we assume the trading only takes place at some given fixed time points $0=t_0<t_1<\ldots<t_N=T$. The trader by the self-financing strategy would hedge at least once on each time period $[t_{i},t_{i+1})$. So, the simplest strategy he can take is to hold the discrete amount
\begin{equation}
	\pi^N_t = \sum_{i=0}^{N-1} \pi_{t_{i}}^N\1_{[t_{i},t_{i+1})}(t),
\end{equation}
of the risky asset $S$ at time $t\ge 0$. Then, the value of the strategy $\pi^N$ is the integral of \eqref{eq:SF_dynamic} minus the proportional transaction cost of it. That is
\begin{equation}\label{eq:kappa}
V^{\pi^N,\kappa}_t = V^{\pi^N,\kappa}_0 + \int_0^t \pi^N_u \, \d S_u - \int_0^t \kappa S_u |\d \pi^N_u|,\notag
\end{equation}
where $\kappa\in (0,1)$ is the proportion of transaction cost.
\begin{dfn}\label{dfn: CMH}
	Consider a European option of the type $f(S_T)$, with convex or concave payoff $f$.
	Let $\pi$ be its Black--Scholes strategy.
	We call the discrete-time strategy $\pi^N$ is a \textit{conditional mean hedging} (CMH) strategy, if for all trading times $t_i$,
	\begin{equation}\label{eq:cm-hedge}
		\E\left[ V^{\pi^N,\kappa}_{t_{i+1}}\,|\, \F_{t_i}\right] = \E\left[ V^{\pi}_{t_{i+1}}\,|\, \F_{t_i}\right],
	\end{equation}
	where $\F_{t_i}$ is the information filter, generated by the asset price process $S$ up to time $t_i$.
\end{dfn}
To extract the strategy $\pi^N$ for such a definition, we need to study those conditional expectations. The following lemmas provide the required conditional laws.
To distinguish the gBm asset model, implied to the BS strategy in the background complete market, from the JD model of the incomplete market, we denote it by $S^{BS}$, and that is
\begin{equation}\label{eq: dS_BS}
	\frac{\d\displaystyle S^{\BS}_t}{\displaystyle S^{\BS}_t}\, =\,\mu\,\d t\,+\,\sigma\,\d W_t,
\end{equation}
with the solution
\begin{equation}\label{eq: S_BS}
	S^{\BS}_t = S^{\BS}_0\,\e^{(\mu - \frac{\sigma^2}{2})t + \sigma W_t},
\end{equation}
and we note for all $t\ge 0$
\begin{equation}\label{eq: JD_BS}
	S_t = S^\BS_t\prod_{j=1}^{N_t}(1+\xi_j).
\end{equation}
\begin{lmm}[Conditional Moments]\label{lmm: Con_Moments}
	For $t\ge u$, and $k\ge 1$ the conditional JD and gBm processes, i.e., $S_t(u) := S_t|\F_u$ and $S^\BS_t(u) := S^\BS_t|\F_u$ has the moments
	\begin{align}
		&\widehat{S^k_t}(u) = \E\left[S^k_t\Big|\F_u\right] = S^k_u\,\exp\left[\Big(k\mu + k(k-1)\frac{\sigma^2}{2} + \lambda\sum_{j=1}^k\binom{k}{j}\epsilon_j\Big)(t-u)\right],\label{eq: hat_Sk}\\
		&\widehat{(S^\BS_t)^k}(u) = \E\left[(S^\BS_t)^k\Big|\F_u\right] = (S^\BS_u)^k\,\exp\left[\Big(k\mu + k(k-1)\frac{\sigma^2}{2}\Big)(t-u)\right],\label{eq: hat_BSk}
	\end{align}
	where $\epsilon_j = \E[\xi^j]$. In particular for $k=1,2$	
	\begin{align}
		\hat S_t(u) &= \E\left[S_t\Big|\F_u\right] = S_u\,\exp\Big[(\mu +\lambda\epsilon_1)(t-u)\Big],\label{eq: hat_S}\\
        \widehat{S^2_t}(u) &= \E\left[S^2_t\Big|\F_u\right] = S^2_u\,\exp\left[\Big(2\mu + \sigma^2 + 2\lambda\epsilon_1 + \lambda\epsilon_2\Big)(t-u)\right].\label{eq: hat_S2}
	\end{align}
\end{lmm}
\begin{proof}
	\begin{align*}
		\widehat{S^k_t}(u) &= \E\left[S^k_t\Big|\F_u\right]\\
		&= S_0^k\,\e^{k(\mu -\frac{\sigma^2}{2})t}\,\E\Big[\e^{k\sigma W_t}\Big|\F_u^W\Big]\,\E\Big[\prod_{j=1}^{N_t}(1+\xi_i)^k\Big|\F_u^{\xi,N}\Big]\\
		&= S_0^k\,\e^{k(\mu -\frac{\sigma^2}{2})t}\,\e^{k\sigma W_u}\,\E\Big[\e^{k\sigma (W_t-W_u)}\Big|\F_u^W\Big]\\
		&\times\prod_{j=1}^{N_u}(1+\xi_j)^k\,\E\Big[\prod_{j=N_u+1}^{N_t}(1+\xi_j)^k\Big|\F_u^{\xi,N}\Big]\\
		&= S_u^k\,\e^{k(\mu -\frac{\sigma^2}{2})(t-u)}\,\E\Big[\e^{k\sigma W_{t-u}}\Big]\,\E\Big[\prod_{j=1}^{N_{t-u}}(1+\xi_j)^k\Big]\\
		&= S_u^k\,\e^{(k\mu + k(k-1)\frac{\sigma^2}{2})(t-u)}\,\sum_{n=0}^\infty\E\Big[\prod_{j=1}^n(1+\xi_j)^k\Big]\,\P\big[N_{t-u}=n\big]\\
		&= S_u\,\e^{(k\mu + k(k-1)\frac{\sigma^2}{2})(t-u)}\,\sum_{n=0}^\infty\frac{\lambda^n(t-u)^n}{n!}\,\e^{-\lambda(t-u)}\left(\E[(1+\xi)^k]\right)^n\\
		&= S_u^k\,\e^{(k\mu + k(k-1)\frac{\sigma^2}{2}-\lambda)(t-u)}\sum_{n=0}^\infty\frac{\lambda^n(t-u)^n}{n!}\left(\sum_{j=0}^k\binom{k}{j}\E[\xi^j]\right)^n\\
		&= S^k_u\,\e^{\left(k\mu + k(k-1)\frac{\sigma^2}{2} + \lambda\sum_{j=1}^k\binom{k}{j}\epsilon_j\right)(t-u)}.
	\end{align*}
\end{proof}
\begin{rem}
	Formula \eqref{eq: hat_S} is a special case of \cite[Theorem 3.1]{almani2024prediction}, for $X=\sigma W+J$. However, the properties of Wiener process increments cause a better closed form for the conditional laws here.
\end{rem}
\begin{lmm}[Conditional Expectation]\label{lmm: Con_f}
	For $t\ge u$, the conditional processes $f^{\BS}_t(u) := f(t,S^{\BS}_t)|\F_u$ and $S.f^{\BS}_t(u):=S_t\cdot f(t,S^{\BS}_t)|\F_u$, we have the following rules
	\begin{align}
		&\label{eq: hat_f}\widehat{f^\BS_t}(u) = \E\left[f(t,S^\BS_t)\Big|\F_u\right]\\
		&= \int_{-\infty}^\infty f\Big(t,S^\BS_u\e^{(\mu - \frac{\sigma^2}{2})(t-u) + \sigma z}\Big)\phi(z;0;t-u)\d z,\notag\\
		&\label{eq: hat_sf}\widehat{S.f^\BS_t}(u) = \E\left[S_t\cdot f(t,S^{\BS}_t)\Big|\F_u\right]\\
		&= S_u\e^{(\mu + \lambda\epsilon_1 - \frac{\sigma^2}{2}) (t-u)}
		\int_{-\infty}^\infty\e^{\sigma z}f\Big(t,S^\BS_u\e^{(\mu - \frac{\sigma^2}{2})(t-u) + \sigma z}\Big)\phi(z;0;t-u)\d z\notag,
	\end{align}
	where $\phi(z;a;b^2)$ is the normal density function of mean $a$ and variance $b^2$, and $G^{*n}$ is the $n$th fold convolution of the distribution $\zeta_k=\log(1+\xi_k)\sim G$.
\end{lmm}
\begin{proof}
	We have
	\begin{align}\label{eq:prf2_u}
		\widehat{f^\BS_t}(u) &= \E\left[f(t,S^\BS_t)\Big|\F_u\right]\notag\\
		&= \E\left[f\Big(t,S^\BS_u\,\e^{(\mu - \frac{\sigma^2}{2})(t-u)\, + \,\sigma (W_t-W_u)}\Big)\Big|\F_u^{W,J}\right]\notag\\
		&= \E\left[f\Big(t,S^\BS_u\,\e^{(\mu - \frac{\sigma^2}{2})(t-u)\, + \,\sigma (W_{t-u})}\Big)\Big|\F_u^W\right].
	\end{align}
	Then we note $S^\BS_u$ and $W_{t-u}\deq W_t-W_u$ are respectively measurable and independent of $\F_u^W$, the expectation is indeed just with respect to the variable $W_{t-u}$. So, using the density function of $W_{t-u}$ which is $\phi(z;0;t-u)$, we can continue \eqref{eq:prf2_u} as
	\begin{equation*}
		= \int_{-\infty}^\infty f\Big(t,S^\BS_u\e^{(\mu - \frac{\sigma^2}{2})(t-u) + \sigma z}\Big)\phi(z;0;t-u)\d z.
	\end{equation*}
	Next, we note $W,N,\xi_k$ are all independent, $W_{t-u}\deq W_t-W_u$ and $N_{t-u}\deq N_t-N_u$ are both independent from $\F_u=\F_u^{W,N,\xi}$, and also $S_u$ and $S^\BS_u$ both are $\F_u$-measurable. So,
		\begin{align*}
		&\widehat{S.f^\BS_t}(u) = \E\left[S_t\cdot f(t,S^{\BS}_t)\Big|\F_u\right]\notag\\
		&= \E\Bigg[S_u\e^{(\mu - \frac{\sigma^2}{2})(t-u) + \sigma (W_t-W_u)}\prod_{j=N_u}^{N_t}(1+\xi_j)\\
		&\times f\Big(t,S^\BS_u\e^{(\mu - \frac{\sigma^2}{2})(t-u) + \sigma (W_t-W_u)}\Big)\Big|\F_u^{W,N,\xi}\Bigg]\notag
		\end{align*}
		\begin{align*}
		&= S_u\e^{(\mu - \frac{\sigma^2}{2})(t-u)}
		\E\left[\prod_{j=1}^{N_{t-u}}(1+\xi_j)\Big|\F_u^{N,\xi}\right]\\
		&\times\E\left[\e^{\sigma W_{t-u}}f\Big(t,S^\BS_u\e^{(\mu - \frac{\sigma^2}{2})(t-u)+\sigma W_{t-u}}\Big)\Big|\F_u^W\right]\\
		&= S_u\e^{(\mu + \lambda\epsilon_1 - \frac{\sigma^2}{2})(t-u)}
		\E\left[\e^{\sigma W_{t-u}}f\Big(t,S^\BS_u\e^{(\mu - \frac{\sigma^2}{2})(t-u)+\sigma W_{t-u}}\Big)\Big|\F_u^W\right],
	\end{align*}
	and again the expectation is just with respect to the variable $W_{t-u}$. This proves \eqref{eq: hat_sf}.
\end{proof}
\section{Conditional Mean Strategies}\label{sect: CMH}
From now on, we consider $\Delta$ as the backward difference operator i.e. $\Delta t_{i+1} = t_{i+1}-t_i$ and for function $f$ it is $\Delta f_{t_{i+1}} = f_{t_{i+1}}-f_{t_i}$.
\begin{lmm}[Conditional Gains]\label{lmm: Gains}
	Let $S$ be left continuous at $\{t_i\}_{i\ge 0}$, i.e. no jump happens right before the transaction payment time points, then
	\begin{align}
		&\hat S_{t_{i+1}}(t_i) =  S_{t_i}\exp\Big[(\mu +\lambda\epsilon_1)\Delta t_{i+1}\Big],\label{eq: hat_Si}\\
		&\widehat{S^2}_{t_{i+1}}(t_i) = S^2_{t_i}\,\exp\left[\Big(2\mu + \sigma^2 + 2\lambda\epsilon_1 + \lambda\epsilon_2\Big)\Delta t_{i+1}\right],\label{eq: hat_S2i}\\
		&\widehat{V}^\pi_{t_{i+1}}(t_i) = \E\left[g(t_{i+1},S^\BS_{t_{i+1}})\Big|\F_{t_i}\right]\label{eq: hat_Vi}\\
		=&
		\int_{-\infty}^\infty g\Big(t_{i+1},S^\BS_{t_i}\e^{(\mu - \frac{\sigma^2}{2})\Delta t_{i+1} + \sigma z}\Big)\phi(z;0;\Delta t_{i+1})\d z\notag\\
		&\widehat{S.V^\pi}_{t_{i+1}}(t_i) = \E\left[S_{t_{i+1}}g(t_{i+1},S^\BS_{t_{i+1}})\Big|\F_{t_i}\right]\label{eq: hat_SVi}\\
		=&\, S_{t_i}\e^{(\mu + \lambda\epsilon_1 - \frac{\sigma^2}{2})\Delta t_{i+1}}
		\int_{-\infty}^\infty\e^{\sigma z}g\Big(t_{i+1},S^\BS_{t_i}\e^{(\mu - \frac{\sigma^2}{2})\Delta t_{i+1} + \sigma z}\Big)\phi(z;0;\Delta t_{i+1})\d z\notag\\
		&V^{\pi^N,\kappa}_{t_{i+1}} = V^{\pi^N,\kappa}_{t_i} + \pi^N_{t_i}\Delta S_{t_{i+1}} - \kappa S_{t_{i+1}}|\Delta\pi^N_{t_{i+1}}|,\label{eq: VNi}\\
		&\hat{V}^{\pi^N,\kappa}_{t_{i+1}}(t_i) = {V}^{\pi^N,\kappa}_{t_i} + \pi^N_{t_i}\Delta\hat S_{t_{i+1}}(t_i) - \kappa \hat S_{t_{i+1}}(t_i)|\Delta\pi^N_{t_{i+1}}|,\label{eq: hat_VNi}
	\end{align}
	where $\Delta\hat S_{t_{i+1}}(t_i) = \hat S_{t_{i+1}}(t_i) - S_{t_i}$.
\end{lmm}
\begin{proof}
	\eqref{eq: hat_Si} and \eqref{eq: hat_S2i} are straight results of Lemma \ref{lmm: Con_Moments}, and formulas \eqref{eq: hat_Vi} and \eqref{eq: hat_SVi} are the results of Lemma \ref{lmm: Con_f}. To prove \eqref{eq: VNi}, we note that from \eqref{eq:kappa}
	\begin{align}
		V^{\pi^N,\kappa}_{t_{i+1}} &= V^{\pi^N,\kappa}_{t_i} + \int_{t_i}^{t_{i+1}} \pi^N_u \, \d S_u - \int_{t_i}^{t_{i+1}} \kappa S_u |\d \pi^N_u|\notag\\
		&= V^{\pi^N,\kappa}_{t_i} + \pi^N_{t_i} (S_{t_{i+1}}-S_{t_i}) - \kappa \lim_{|\Lambda_n|\to 0}\sum_{j=0}^{n-1}S_{u^*_j} |\Delta \pi^N_{u_{j+1}}|\label{eq: V_lim},
	\end{align}
where $\Lambda_n: t_i=u_0, u_1, \ldots, u_{n-1}, u_n=t_{i+1}$ is a partition of $n+1$ time points on $[t_i,t_{i+1}]$, that $u^*_j\in [u_j,u_{j+1}]$, and $|\Lambda_n|=\max_{0\le j\le n-1}|\Delta u_j|$. Now, as the asset price $S$ is left continuous at $\{t_i\}_{i\ge 0}$, we can continue \eqref{eq: V_lim} as folloing
\begin{equation*}
	= V^{\pi^N,\kappa}_{t_i} + \pi^N_{t_i} (S_{t_{i+1}}-S_{t_i}) - \kappa S_{t_{i+1}} |\Delta \pi^N_{t_{i+1}}|.
\end{equation*}
Taking the conditional expectation $\E[\,\cdot\,|\F_{t_i}]$ to the right side of the final equality above, we have \eqref{eq: hat_VNi}.
\end{proof}
\begin{thm}[Conditional Mean Hedging]\label{thm: CMH}
	If the asset price $S$ is left continuous, $\pi^N$ is the CMH strategy for the European option of type $f(S_T)$ with convex or concave positive payoff function $f$, and transaction costs proportion $\kappa$, if and only if
	\begin{equation}\label{eq: recursive}
		|\pi^N_{t_{i+1}}-\pi^N_{t_i}|=\frac{V^{\pi^N,\kappa}_{t_i} - \hat{V}^\pi_{t_{i+1}}(t_i) + \pi^N_{t_i}\Delta\hat S_{t_{i+1}}(t_i)}{\kappa\hat S_{t_{i+1}}(t_i)},
	\end{equation}
	for all $i=0,\ldots,n-1$.
\end{thm}
\begin{proof}
	The definition \ref{dfn: CMH} and equation \eqref{eq:cm-hedge}, are equivalent to have \eqref{eq: hat_Vi} and \eqref{eq: hat_VNi} equal. This returns the equation \eqref{eq: recursive}.
\end{proof}
\begin{rem}
	Taking a long position (buying) of $S$ for the period $[t_i,t_{i+1})$ means $\Delta\pi^N_{t_i}>0$, and so from \eqref{eq: recursive}
	\begin{equation}\label{eq: Long}
		\pi^N_{t_{i+1}} = \pi^N_{t_i}\, +\, \frac{V^{\pi^N,\kappa}_{t_i} - \hat{V}^\pi_{t_{i+1}}(t_i) + \pi^N_{t_i}\Delta\hat S_{t_{i+1}}(t_i)}{\kappa\hat S_{t_{i+1}}(t_i)}.
	\end{equation}
	On the other hand, a short position (selling) of $S$ for the period $[t_i,t_{i+1})$ means $\Delta\pi^N_{t_i}<0$, hence from \eqref{eq: recursive}
	\begin{equation}\label{eq: Short}
		\pi^N_{t_{i+1}} = \pi^N_{t_i}\, -\, \frac{V^{\pi^N,\kappa}_{t_i} - \hat{V}^\pi_{t_{i+1}}(t_i) + \pi^N_{t_i}\Delta\hat S_{t_{i+1}}(t_i)}{\kappa\hat S_{t_{i+1}}(t_i)}.
	\end{equation}
	These cause a binary decision tree of strategies that a trader can take on transaction time intervals as follows.
	\begin{figure}[H]
		\begin{forest}
			for tree={
				calign=center,
				grow'=east, 
				text height=1.4ex, text depth=0.2ex, 
				l sep+=4em,
				inner sep=0.25em
			}
			[$\pi_{t_0}^{N}$
			[$\pi_{t_1}^{N,1}$,edge label={node[midway,above,sloped,font=\it\small]{long}}
			[$\pi_{t_2}^{N,1}$,edge label={node[midway,above,sloped,font=\it\small]{long}}
			[$\pi_{t_3}^{N,1}$
			[$\pi_{t_{N-1}}^{N,1}$, edge=dashed][,no edge]]
			[$\pi_{t_3}^{N,2}$[{\LARGE :},no edge]]]
			[$\pi_{t_2}^{N,2}$,edge label={node[midway,below,sloped,font=\it\small]{short}}
			[$\pi_{t_3}^{N,3}$[{\LARGE :},no edge]]
			[$\pi_{t_3}^{N,4}$[{\LARGE :},no edge]]]]
			[$\pi_{t_1}^{N,2}$,edge label={node[midway,below,sloped,font=\it\small]{short}}
			[$\pi_{t_2}^{N,3}$,edge label={node[midway,above,sloped,font=\it\small]{long}}
			[$\pi_{t_3}^{N,5}$[{\LARGE :},no edge]]
			[$\pi_{t_3}^{N,6}$[{\LARGE :},no edge]]]
			[$\pi_{t_2}^{N,4}$,edge label={node[midway,below,sloped,font=\it\small]{short}}
			[$\pi_{t_3}^{N,7}$[{\LARGE :},no edge]]
			[$\pi_{t_3}^{N,8}$ [,no edge]
			[$\pi_{t_{N-1}}^{N,2^{N-1}}$, edge=dashed]]
			]]]
		\end{forest}
		\caption{The decision tree of CMH method.}
	\end{figure}
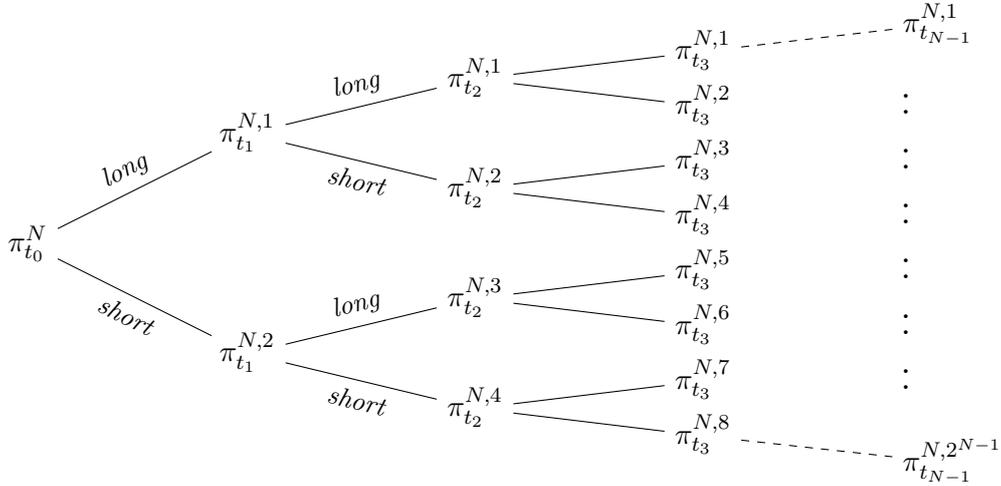
\end{rem}
\section{Conditional Least-Square Strategies}\label{sect: CLH}
\begin{rem}[Interpretations]
	Although Theorem \ref{thm: CMH} results in an explicit form for the CMH strategy, there are some considerable points in it:
	
	1. Denote
	\begin{align}
		\text{(Empirical)}\quad&\theta^N_{t_i} = V^{\pi^N,\kappa}_{t_i} + \pi_{t_i}^N \Delta\hat S_{t_{i+1}}(t_i),\label{eq: theta_N}\\
		\text{(Theoretical)}\quad&\theta^\pi_{t_i} = \hat V^\pi_{t_{i+1}}(t_i).\label{eq: theta_p}
	\end{align}
	Theorem \ref{thm: CMH} shows the CMH strategy is plausible only if $\theta^N_{t_i} - \theta^\pi_{t_i}\ge 0$, i.e., the empirical strategy's values always exceed the theoretical values for $i=0,\dots,n-1$ all ($\theta^N_{t_i} \ge \theta^\pi_{t_i}$). However, this theorem vanishes if just for some $0\le j\le n-1$ we happen to $\theta^N_{t_j} < \theta^\pi_{t_j}$.
	
	2. On the other hand, even if we successfully extract the CMH strategy, that is matching the conditional means of the variables $V^{\pi^N,\kappa}_{t_{i+1}}$ and $V^{\pi}_{t_{i+1}}$. However, this strategy does not guarantee the minimum conditional distance of these variables necessarily.
	
	3. Since under the transactional costs the perfect hedging is not possible. However, we note that applying the JD model \eqref{eq:JD} for the risky asset price, is indeed a generalization of the gBm model \eqref{eq:gBm} by adding a jump to it. So, for a trader in JD incomplete market, it is natural to investigate the closest discrete strategy to the perfect hedging strategy in the background gBm complete market, i.e., the BS strategy. Now, if one consider the difference (distance) as the conditional least-square norm, then the mathematical formulation of this methodology for hedging forms the following definition as a modified and alternative framework instead of the CMH.
\end{rem}
\begin{dfn}[Conditional Least Square Hedging]\label{dfn: CLH}
	Let $f(S_T)$ be a financial derivative with convex or concave payoff $f$.
	Let $\pi$ be its Black--Scholes strategy.
	We call the discrete-time strategy $^*\pi^N$ is a \textit{conditional least square hedging} (CLH) strategy, if for all trading times $t_i$,
	\begin{equation}\label{eq:cl-hedge}
		^*\pi^N_{t_{i+1}}\in\underset{\pi^N_{t_{i+1}}\in\R}{\mathrm{argmin}}\;\E\left[\left( V^{\pi^N,\kappa}_{t_{i+1}} - V^{\pi}_{t_{i+1}}\right)^2\,\Big|\, \F_{t_i}\right].
	\end{equation}
	Here $\F_{t_i}$ is the information filter, generated by the asset price $S$ upto time $t_i$.
\end{dfn}
\begin{lmm}\label{lmm: opt}
	Consider the problem
	\begin{equation}\label{eq: opt}
		\min_{x\in\R}\,f(x) = a(x-x_0)^2 + b|x-x_0| + c,
	\end{equation}
	where $a> 0$, $b,x_0\in\R$, and $c\ge 0$.\\
	i) If $b\ge 0$, it has the unique solution $x^* = x_0$ with $f(x^*)=c$,\\
	ii) If $b<0$, it has the two solutions $x^* = x_0\pm b/2a$ with $f(x^*)=0$.
\end{lmm}
\begin{proof}
	First, if $b\ge 0$ then $f$ is a convex function which has its unique minimum on $x^*=x_0$ with $f(x^*)=c$. Second, if $b<0$, we have
	\begin{equation}
		f(x) =
		\left\{
		\begin{array}{lll}
			a(x-x_0)^2 - b(x-x_0) + c &;&x<x_0\\
			c &;&x=x_0\\
			a(x-x_0)^2 + b(x-x_0) + c &;&x>x_0
		\end{array}
		\right.
	\end{equation}
	and so
	\begin{equation}
		f'(x) =
		\left\{
		\begin{array}{lll}
			2a(x-x_0) - b &;&x<x_0\\
			2a(x-x_0) + b &;&x>x_0
		\end{array}
		\right.
	\end{equation}
	and $f''(x)\equiv 2a>0$ on all $x\in\R$. So, $f$ has minimums around the roots of its derivative $f'$, which are $x^*=x_0\pm b/2a$ with values $f(x^*)=0$.
\end{proof}
\begin{thm}[CLH Strategy]\label{thm: CLH}
	Denote
	\begin{equation}\label{eq: test_U}
		U_{t_i} = \left(V^{\pi^N,\kappa}_{t_i} - \pi_{t_i}^N S_{t_i}\right)\hat S_{t_{i+1}}(t_i) - \widehat{S.V^{\pi}}_{t_{i+1}}(t_i)
		+\pi_{t_i}^N\widehat{S^2}_{t_{i+1}}(t_i).
	\end{equation}
	If the asset price $S$ is left continuous, for the European  option of the type $f(S_T)$ with convex or concave positive payoff function $f$, and proportion of transaction costs $\kappa$ the CLH strategy admits the following recursive equations:\\
	(I) If $U_{t_i}\le 0$, then $^*\pi^N_{t_{i+1}}
	= \pi^N_{t_i}$,\\
	(II) If $U_{t_i}> 0$, then for a long position on $t_{i+1}$
	\begin{align}\label{eq: long}
		^*\pi^N_{t_{i+1}}(long) &= \pi^N_{t_i} + \frac{U_{t_i}}{\kappa\widehat{S^2}_{t_{i+1}}(t_i)}\\
		&= \pi^N_{t_i} + \frac{1}{\kappa}\left[\pi^N_{t_i} - \frac{\widehat{\,S.V^\pi \,}_{t_{i+1}}(t_i) - (V^{\pi^N,\kappa}_{t_i}-\pi^N_{t_i}S_{t_i})\hat S_{t_{i+1}}(t_i)}{\widehat{S^2}_{t_{i+1}}(t_i)}\right],\notag
	\end{align}
    and for a short position on $t_{i+1}$
	\begin{align}\label{eq: short}
		^*\pi^N_{t_{i+1}}(short) &= \pi^N_{t_i} - \frac{U_{t_i}}{\kappa\widehat{S^2}_{t_{i+1}}(t_i)}\\	
		&= \pi^N_{t_i} - \frac{1}{\kappa}\left[\pi^N_{t_i} - \frac{\widehat{\,S.V^\pi \,}_{t_{i+1}}(t_i) - (V^{\pi^N,\kappa}_{t_i}-\pi^N_{t_i}S_{t_i})\hat S_{t_{i+1}}(t_i)}{\widehat{S^2}_{t_{i+1}}(t_i)}\right].\notag
	\end{align}
\end{thm}
\begin{proof}
	By Lemma \ref{lmm: Gains}
	\begin{align*}
		V^{\pi^N,\kappa}_{t_{i+1}} - V^{\pi}_{t_{i+1}}
		&= V^{\pi^N,\kappa}_{t_{i}} - V^{\pi}_{t_{i+1}}
		+\pi_{t_i}^N(S_{t_{i+1}} - S_{t_i})
		-\kappa S_{t_{i+1}}\left|\pi_{t_{i+1}}^N - \pi_{t_i}^N\right|,
	\end{align*}
	and so
	\begin{align*}
		\left(V^{\pi^N,\kappa}_{t_{i+1}} - V^{\pi}_{t_{i+1}}\right)^2
		&= \kappa^2 S^2_{t_{i+1}}\left(\pi_{t_{i+1}}^N - \pi_{t_i}^N\right)^2\\
		&-2\kappa S_{t_{i+1}}\left\{V^{\pi^N,\kappa}_{t_{i}} - V^{\pi}_{t_{i+1}}
		+\pi_{t_i}^N(S_{t_{i+1}} - S_{t_i})\right\}\left|\pi_{t_{i+1}}^N - \pi_{t_i}^N\right|\\
		&+\left(V^{\pi^N,\kappa}_{t_{i}} - V^{\pi}_{t_{i+1}}
		+\pi_{t_i}^N(S_{t_{i+1}} - S_{t_i})\right)^2,
	\end{align*}
	and taking the conditional expectation we have
		\begin{align*}
		&\E\left[\left(V^{\pi^N,\kappa}_{t_{i+1}} - V^{\pi}_{t_{i+1}}\right)^2\Big|\F_{t_i}\right]
		= \kappa^2\widehat{S^2}_{t_{i+1}}\left(\pi_{t_{i+1}}^N - \pi_{t_i}^N\right)^2\\
		&-2\kappa\left\{\left(V^{\pi^N,\kappa} - \pi_{t_i}^N S_{t_i}\right)\hat S_{t_{i+1}}(t_i) - \widehat{S.V^{\pi}}_{t_{i+1}}(t_i)
		+\pi_{t_i}^N\widehat{S^2}_{t_{i+1}}(t_i)\right\}\left|\pi_{t_{i+1}}^N - \pi_{t_i}^N\right|\\
		&+\E\left[\left(V^{\pi^N,\kappa}_{t_{i}} - V^{\pi}_{t_{i+1}}
		+\pi_{t_i}^N(S_{t_{i+1}} - S_{t_i})\right)^2\Big|\F_{t_i}\right]\\
		&= a_i\left(\pi_{t_{i+1}}^N - \pi_{t_i}^N\right)^2 + b_i\left|\pi_{t_{i+1}}^N - \pi_{t_i}^N\right| + c_i,
	\end{align*}
	and this is a function of the form \eqref{eq: opt}. So, by the Lemma \ref{lmm: opt}, if $b_i\ge 0$ ($U_i\le 0$) then it has a unique minimum at $^*\pi_{t_{i+1}}^N = \pi_{t_i}^N$. If $b_i< 0$ ($U_i> 0$) then it has twople minimums at $^*\pi_{t_{i+1}}^N = \pi_{t_i}^N \pm b_i/2a_i = \pi_{t_i}^N \mp U_i/\kappa\widehat{S^2}_{t_{i+1}}(t_i).$
\end{proof}
\begin{rem}
	Similar to the CMH method, here also we will face a decision tree of the optimal strategies of CLH method. However, the decision tree of CLH method is different. In some branches that strategy does not change, the branch continues straight but not in a binary shape. In other words, we face a decision tree similar to Figure \ref{Tree_CLH}.
		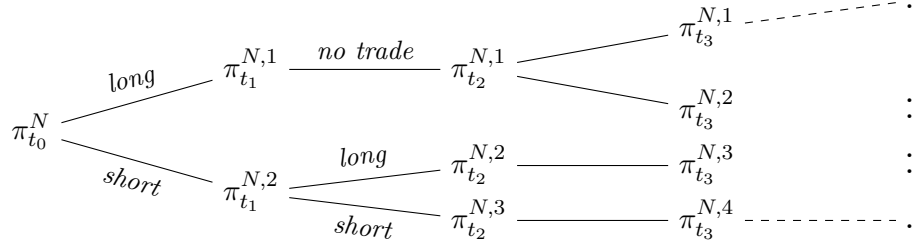
\begin{figure}[H]
		\begin{forest}
			for tree={
				calign=center,
				grow'=east, 
				text height=1.4ex, text depth=0.2ex, 
				l sep+=4em,
				inner sep=0.25em
			}
			[$\pi_{t_0}^{N}$
			[$\pi_{t_1}^{N,1}$,edge label={node[midway,above,sloped,font=\it\small]{long}}
			[$\pi_{t_2}^{N,1}$,edge label={node[midway,above,sloped,font=\it\small]{no trade}}
			[$\pi_{t_3}^{N,1}$
			[{\LARGE .}, edge=dashed][,no edge]]
			[$\pi_{t_3}^{N,2}$[{\LARGE :},no edge]]]
			]
			[$\pi_{t_1}^{N,2}$,edge label={node[midway,below,sloped,font=\it\small]{short}}
			[$\pi_{t_2}^{N,2}$,edge label={node[midway,above,sloped,font=\it\small]{long}}
			[$\pi_{t_3}^{N,3}$[{\LARGE :},no edge]]
			]
			[$\pi_{t_2}^{N,3}$,edge label={node[midway,below,sloped,font=\it\small]{short}}
			[$\pi_{t_3}^{N,4}$[{\LARGE .}, edge=dashed]]
			]]]
		\end{forest}
		\caption{The decision tree of CLH method.}\label{Tree_CLH}
	\end{figure}
\end{rem}
\begin{rem}\label{rem: CLH_test}
	Denote
	\begin{align}
		\text{(Empirical)}\quad&\ell^N_{t_i} = \left(V^{\pi^N,\kappa}_{t_i} - \pi_{t_i}^N S_{t_i}\right)\hat S_{t_{i+1}}(t_i)
		+\pi_{t_i}^N\widehat{S^2}_{t_{i+1}}(t_i),\label{eq: ell_N}\\
		\text{(Theoretical)}\quad&\ell^\pi_{t_i} = \widehat{S.V^{\pi}}_{t_{i+1}}(t_i).\label{eq: ell_p}
	\end{align}
	As $U_{t_i} = \ell^N_{t_i} - \ell^\pi_{t_i}$, further than the explicit strategies, Theorem \ref{thm: CLH} reveals  a test. In other words, it indicates that if the empirical value remains less or equal to the theoretical value ($\ell^N_{t_i}\leq\ell^\pi_{t_i}$), the trader does not need to change the volume of the strategy to stay close to the BS strategy's value. However, if the empirical value exceeds the theoretical value ($\ell^N_{t_i}>\ell^\pi_{t_i}$), in order to stay close to BS strategy's value, the trader should update the volume of the strategy according to the formulas \eqref{eq: long} and \eqref{eq: short}.
\end{rem}
\section{European Vanilla Call Option}\label{sect: EU_call}
Next, we are interested in calculating the European Call option case. To do this, first we need the following lemma. From now on, we denote the density, cumulative probability, and the expectation functions of the normal standard distribution respectively by $\varphi,\,\Phi,\,\E_z$, and the expectation function of the normal distribution $\mathcal{N}(\mu,\sigma^2)$ by $\E_z^{\mu,\sigma^2}$, i.e., for all $f:\R\to\R$
\begin{align*}
	&\varphi(z) = \frac{1}{\sqrt{2\pi}}\,\e^{-\frac{z^2}{2}},\qquad
	\Phi(z) = \int_{-\infty}^z\varphi(u)\,\d u = \frac{1}{\sqrt{2\pi}}\int_{-\infty}^z\e^{-\frac{u^2}{2}}\,\d u,\\
	&\E_z\big[f(z)\big] = \int_{-\infty}^\infty f(u)\,\varphi(u)\,\d u = \frac{1}{\sqrt{2\pi}}\int_{-\infty}^\infty f(u)\,\e^{-\frac{u^2}{2}}\,\d u,\\
	&\E_z^{\mu,\sigma^2}\big[f(z)\big] = \int_{-\infty}^\infty f(u)\,\phi(u;\mu;\sigma^2)\,\d u = \frac{1}{\sqrt{2\pi\sigma^2}}\int_{-\infty}^\infty f(u)\,\e^{-\frac{(u-\mu)^2}{2\sigma^2}}\,\d u,
\end{align*}
providing the right side integrals exist.
\begin{lmm}\label{lmm: Eint}
	For all given real values $a,b,\alpha,\mu\in\R$ and $\sigma >0$
	\begin{equation}\label{eq: int_Phi}
		\int_{-\infty}^{\infty}\e^{\alpha z}\,\Phi(az+b)\,\phi(z;\mu;\sigma^2)\,\d z\, =\, \e^{\alpha\mu + \frac{\alpha^2\sigma^2}{2}}\Phi\Bigg(\frac{a(\mu+\alpha\sigma^2) + b}{\displaystyle\sqrt{1+a^2\sigma^2}}\Bigg),
	\end{equation}
	in other word
	\begin{equation}
		\E_z^{\mu,\sigma^2}\Big[\e^{\alpha z}\,\Phi(az+b)\Big] = \e^{\alpha\mu + \frac{\alpha^2\sigma^2}{2}}\Phi\Bigg(\frac{a(\mu+\alpha\sigma^2) + b}{\displaystyle\sqrt{1+a^2\sigma^2}}\Bigg).
	\end{equation}
	In particular
	\begin{gather}
		\E_z\Big[\e^{\alpha z}\,\Phi(az+b)\Big] = \int_{-\infty}^{\infty}\e^{\alpha z}\,\Phi(az+b)\,\varphi(z)\,\d z\, =\, \e^{\frac{\alpha^2}{2}}\Phi\Bigg(\frac{a\alpha + b}{\displaystyle\sqrt{1+a^2}}\Bigg),\\
		\E_z^{\mu,\sigma^2}\Big[\Phi(az+b)\Big] = \Phi\Bigg(\frac{a\mu + b}{\displaystyle\sqrt{1+a^2\sigma^2}}\Bigg).
	\end{gather}
\end{lmm}
\begin{proof}
	By changing the variable $u = \frac{z-\mu}{\sigma}$ we have
	\begin{align}
		&\int_{-\infty}^{\infty}\e^{\alpha z}\,\Phi(az+b)\,\phi(z;\mu;\sigma^2)\,\d z\notag\\
		&=\frac{1}{\sqrt{2\pi\sigma^2}}\int_{-\infty}^\infty \e^{\alpha z}\,\Phi(az+b)\,\e^{-\frac{(z-\mu)^2}{2\sigma^2}}\,\d z\notag\\
		&=\frac{\e^{\alpha\mu}}{\sqrt{2\pi}}\int_{-\infty}^\infty \Phi(a\sigma u + a\mu + b)\,\e^{\alpha\sigma u -\frac{u^2}{2}}\,\d u\notag\\
		&=\frac{\e^{\alpha\mu + \frac{\alpha^2\sigma^2}{2}}}{\sqrt{2\pi}}\int_{-\infty}^\infty \Phi(a\sigma u + a\mu + b)\,\e^{-\frac{(u-\alpha\sigma)^2}{2}}\,\d u\notag\\
		&=\e^{\alpha\mu + \frac{\alpha^2\sigma^2}{2}}\int_{-\infty}^\infty \Phi(a\sigma u + a\mu + b)\,\phi(u;\alpha\sigma;1)\,\d u.\label{eq: int1}
	\end{align}
	Next, we not $\Phi(z) = \frac12\left\{1+\erf\left(\frac{z}{\sqrt2}\right)\right\}$ for $\erf(z) = \frac{2}{\sqrt{\pi}}\int_0^z\e^{-u^2}\,\d u$. So, one can continue \eqref{eq: int1} following
	\begin{align}
		&=\e^{\alpha\mu + \frac{\alpha^2\sigma^2}{2}}\int_{-\infty}^\infty \frac12\left\{1+\erf\left(\frac{a\sigma u + a\mu + b}{\sqrt2}\right)\right\}\,\phi(u;\alpha\sigma;1)\,\d u\notag\\
		&=\frac{\e^{\alpha\mu + \frac{\alpha^2\sigma^2}{2}}}{2}\left\{1 + \int_{-\infty}^\infty\erf\left(\frac{a\sigma u + a\mu + b}{\sqrt2}\right)\,\phi(u;\alpha\sigma;1)\,\d u\right\}.\label{eq: int2}
	\end{align}
	Now, from \cite[Section 4.3, Formula 13]{ng1969table}, for arbitrary $A,B,m\in\R$ and $\nu > 0$
	$$\int_{-\infty}^{\infty}\erf(Au+B)\,\phi(u;m;\nu^2)\,\d u = \erf\left(\frac{Am + B}{\displaystyle\sqrt{1+2A^2\nu^2}}\right).$$
	So, we can continue \eqref{eq: int2} as following
	\begin{align*}
		&=\frac{\e^{\alpha\mu + \frac{\alpha^2\sigma^2}{2}}}{2}\left\{1 + \erf\left(\frac{a\alpha\sigma^2 + a\mu + b}{\displaystyle\sqrt{1+a^2\sigma^2}}\Big/\sqrt2\right)\right\}\\
		&=\e^{\alpha\mu + \frac{\alpha^2\sigma^2}{2}}\Phi\Bigg(\frac{a(\mu+\alpha\sigma^2) + b}{\displaystyle\sqrt{1+a^2\sigma^2}}\Bigg).
	\end{align*}
\end{proof}
\begin{cor}[European Call Option]\label{cor: EU_Call}
	For the European call option with left continuous underlying asset price $S$ and strike-price $K$, denote
	\begin{gather*}
		a_i = \frac{\sigma\Delta t_{i+1}}{\displaystyle\sqrt{T-t_{i+1}}},\quad
		c_i = \displaystyle\sqrt{\frac{\Delta t_{i+1}}{T-t_{i+1}}},\quad
		b^-_i = b^+_i - \sigma\displaystyle\sqrt{T-t_{i+1}},\\
		b_i^+ = \frac{\ln\frac{S^\BS_{ti}}{K} + (\mu - \frac{\sigma^2}{2})\Delta t_{i+1} + \frac{\sigma^2}{2}(T-t_{i+1})}{\displaystyle\sigma\sqrt{T-t_{i+1}}},
	\end{gather*}
	then
	\begin{gather}\label{eq: ell_Call}
		\ell_{t_i}^{\Call} = \widehat{S.V}^{\Call}_{t_{i+1}}(t_i)\\
		= S_{t_i}\e^{(\mu + \lambda\epsilon_1)\Delta t_{i+1}}
		\left[S^\BS_{t_i}\e^{(\mu + \sigma^2)\Delta t_{i+1}}
		\Phi\left(\frac{2a_i + b^+_i}{\displaystyle\sqrt{1+c^2_i}}\right)
		-
		K\Phi\left(\frac{a_i + b^-_i}{\displaystyle\sqrt{1+c^2_i}}\right)\right],\notag
	\end{gather}
	and if $\ell_{t_i}^N\le\ell_{t_i}^{\Call}$, the CLH strategy admits $^*\pi^N_{t_{i+1}}=\pi^N_{t_i}$. If $\ell_{t_i}^N>\ell_{t_i}^{\Call}$ then
	\begin{equation}\label{eq: LS_Call}
	    ^*\pi^N_{t_{i+1}}(long) = \pi^N_{t_i} + \frac{\ell_{t_i}^N-\ell_{t_i}^{\Call}}{\kappa\widehat{S^2}_{t_{i+1}}(t_i)},\quad
	    ^*\pi^N_{t_{i+1}}(short) = \pi^N_{t_i} - \frac{\ell_{t_i}^N-\ell_{t_i}^{\Call}}{\kappa\widehat{S^2}_{t_{i+1}}(t_i)}.
	\end{equation}
\end{cor}
\begin{proof}
	For the European call option with zero rate riskless asset we have the Black--Scholes solution to the equation \eqref{eq:BS} is
	\begin{gather*}
		g_\Call(t,S^\BS_t) = S^\BS_t\Phi(d^+_t) - K\Phi(d^-_t),\\
		d^+_t = \frac{\ln\frac{S^\BS_t}{K}+\frac{\sigma^2}{2}(T-t)}{\sigma\sqrt{T-t}},\\
		d^-_t = d^+_t - \sigma\sqrt{T-t}.
	\end{gather*}
	So by Lemma \ref{lmm: Gains}
	\begin{align}
		\ell_{t_i}^{\Call} &= \widehat{S.V}^{\Call}_{t_{i+1}}(t_i) = \E\left[S_{t_{i+1}}g_\Call(t_{i+1},S^\BS_{t_{i+1}})\Big|\F_{t_i}\right]\\
		&= S_{t_i}\e^{(\mu + \lambda\epsilon_1 - \frac{\sigma^2}{2})\Delta t_{i+1}}\notag\\
		&\times\int_{-\infty}^\infty\e^{\sigma z}g_\Call(t_{i+1},S^\BS_{t_{i+1}})\phi(z;0;\Delta t_{i+1})\d z\notag\\
		&= S_{t_i}\e^{(\mu + \lambda\epsilon_1 - \frac{\sigma^2}{2})\Delta t_{i+1}}\notag\\
		&\times\int_{-\infty}^\infty\e^{\sigma z}\Big(S_{t_{i+1}}\Phi(d^+_{t_{i+1}}) - K\Phi(d^-_{t_{i+1}})\Big)\phi(z;0;\Delta t_{i+1})\d z.\label{eq: call_1}
	\end{align}
	Here the inner integral is
	\begin{align*}
		&\E_z^{0,\Delta t_{i+1}}\left[\e^{\sigma z}\Big(S^\BS_{t_{i+1}}\Phi(d^+_{t_{i+1}}) - K\Phi(d^-_{t_{i+1}})\Big)\right]\\
		&= S^\BS_{t_i}\e^{(\mu - \frac{\sigma^2}{2})\Delta t_{i+1}}\;\E_z^{0,\Delta t_{i+1}}\left[\e^{2\sigma z}\Phi\left(\frac{\ln\frac{S^\BS_{t_i}\e^{(\mu - \frac{\sigma^2}{2})\Delta t_{i+1} + \sigma z}}{K}+\frac{\sigma^2}{2}(T-t_{i+1})}{\displaystyle\sigma\sqrt{T-t_{i+1}}}\right)\right]\\
		&- K\;\E_z^{0,\Delta t_{i+1}}\left[\e^{\sigma z}\Phi\left(\frac{\ln\frac{S^\BS_{t_i}\e^{(\mu - \frac{\sigma^2}{2})\Delta t_{i+1} + \sigma z}}{K}-\frac{\sigma^2}{2}(T-t_{i+1})}{\displaystyle\sigma\sqrt{T-t_{i+1}}}\right)\right]\\
		&= S^\BS_{t_i}\e^{(\mu - \frac{\sigma^2}{2})\Delta t_{i+1}}\;\E_z^{0,\Delta t_{i+1}}\left[\e^{2\sigma z}\Phi\left(\frac{\scriptstyle{\ln\frac{S^\BS_{t_i}}{K}+{(\mu-\frac{\sigma^2}{2})\Delta t_{i+1}+\sigma z}+\frac{\sigma^2}{2}(T-t_{i+1})}}{\displaystyle\sigma\sqrt{T-t_{i+1}}}\right)\right]\\
		&- K\;\E_z^{0,\Delta t_{i+1}}\left[\e^{\sigma z}\Phi\left(\frac{\scriptstyle{\ln\frac{S^\BS_{t_i}}{K}+{(\mu-\frac{\sigma^2}{2})\Delta t_{i+1}+\sigma z}-\frac{\sigma^2}{2}(T-t_{i+1})}}{\displaystyle\sigma\sqrt{T-t_{i+1}}}\right)\right],
	\end{align*}
	applying Lemma \ref{lmm: Eint} we can continue
	\begin{equation*}
		= S^\BS_{t_i}\e^{(\mu + \frac{3}{2}\sigma^2)\Delta t_{i+1} + y}\Phi\left(\frac{2a_i + b_i^+}{\displaystyle\sqrt{1 + c_i^2}}\right) - K\e^{\frac{\sigma^2}{2}\Delta t_{i+1}}\Phi\left(\frac{a_i + b_i^-}{\displaystyle\sqrt{1 + c_i^2}}\right).
	\end{equation*}
	Now, by substituting this to \eqref{eq: call_1} proves \eqref{eq: ell_Call},
	and applying this result to the Theorem \ref{thm: CLH} and Remark \ref{rem: CLH_test} proves the rest of this corollary.
\end{proof}
\section{Simulations}\label{sect: simulations}
In this section, we simulate the result of Corollary \ref{cor: EU_Call}. First we consider the case $\xi\equiv -0.5$, i.e., pure negative Poisson jumps in Figure \ref{fig1} and its decision tree in Figure \ref{fig1_1}. Then, we consider the case $\xi\equiv 0.5$, the pure positive Poisson jumps and its results are given in Figures \ref{fig2} and \ref{fig2_1}. The simulations are done for $T=12$ months (1 year), including 5 transaction payment times, with $\kappa = 0.1$. In both simulations we consider $\mu = 0.15, \sigma = 0.25$ and $\lambda = 0.3$ per month.
\begin{rem}\label{rem: reason}
	Considering the simulations and figures, one can see in branches that the optimal strategy $^*\pi^N_{t_{i+1}}$ changes, it is diverging from $\pi^N_{t_i}$ rapidly. To explain the reason of this, we must note to the equations \eqref{eq: VNi} and \eqref{eq:cl-hedge}. Indeed, we aim to minimize the difference of
	\[V^{\pi^N,\kappa}_{t_{i+1}} = V^{\pi^N,\kappa}_{t_i} + \pi^N_{t_i}(S_{t_{i+1}}-S_{t_i}) - \kappa S_{t_{i+1}}|\pi^N_{t_{i+1}}-\pi^N_{t_i}|,\]
	and $V^{\pi}_{t_{i+1}}$, with respect to the overall information up to time $t_i$, i.e., $\F_{t_i}$. While $V^{\pi}$ changes slightly (with no jump) from $t_i$ to $t_{i+1}$, the JD asset price $S$ changes roughly high in that time period if some jumps happen. In other word, the part $\pi^N_{t_i}(S_{t_{i+1}}-S_{t_i})$ is roughly high when some jumps happen on $[t_i,t_{i+1})$. So, the methodology of CLH tries to compensate this by the part $-\kappa S_{t_{i+1}}|\pi^N_{t_{i+1}}-\pi^N_{t_i}|$. To do this, it takes a value for $^*\pi^N_{t_{i+1}}$ extremely different from $\pi^N_{t_i}$, to overcome the reducing effect of proportion $\kappa$ as well as the roughness of the part $\pi^N_{t_i}(S_{t_{i+1}}-S_{t_i})$. In other word, the jumps in optimal strategy values are the {\it ``cost"} of overcoming to the jumps in the underlying asset price.
\end{rem}
\begin{figure}[H]
	\includegraphics[width=0.95\textwidth]{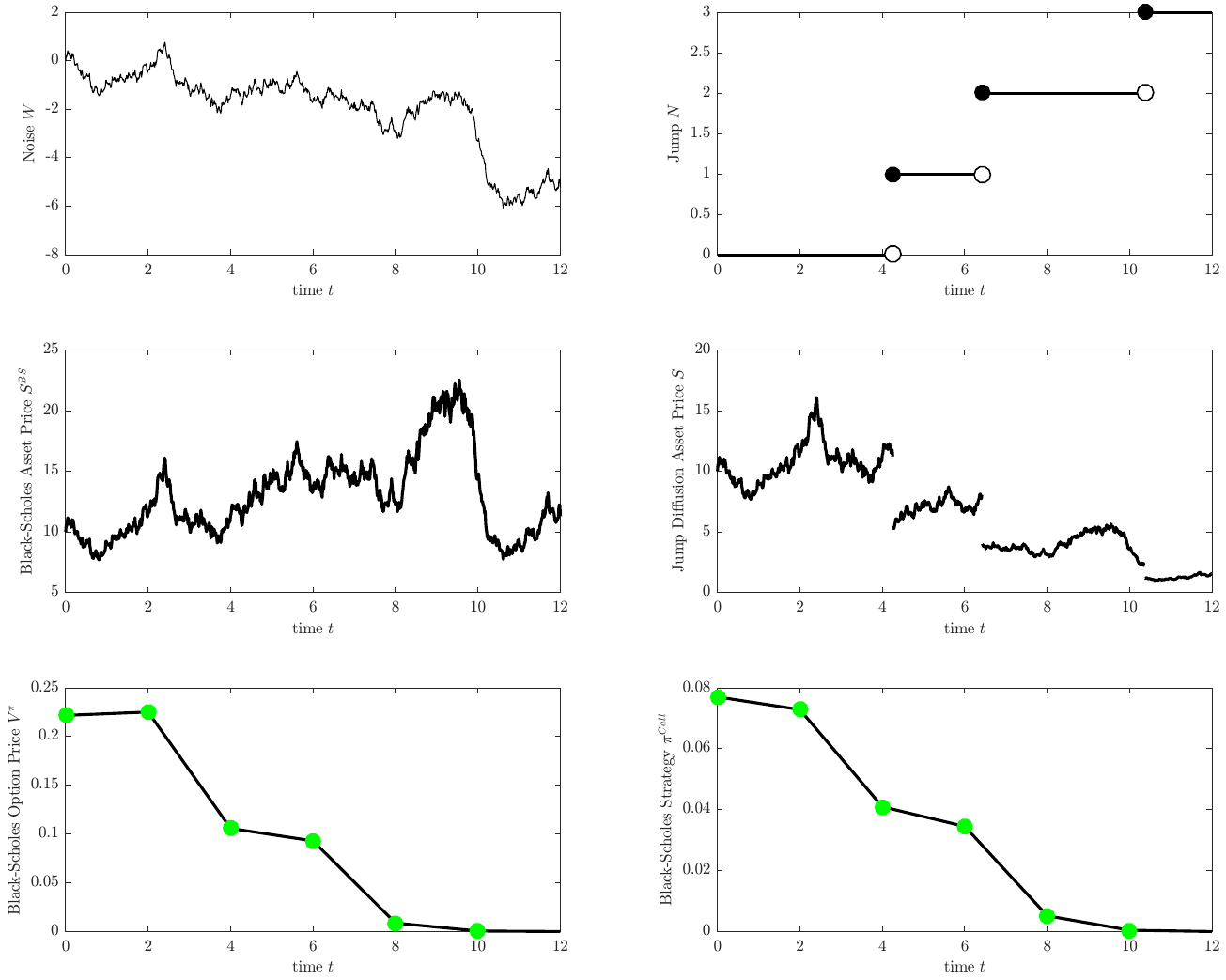}
	\caption{The underlying asset price, European call option price and strategy, the noise and jumps for constants $\mu = 0.15, \sigma = 0.25, \xi\equiv -0.5, T=12, \kappa = 0.1$ and $\lambda = 0.3$.}\label{fig1}
\end{figure}
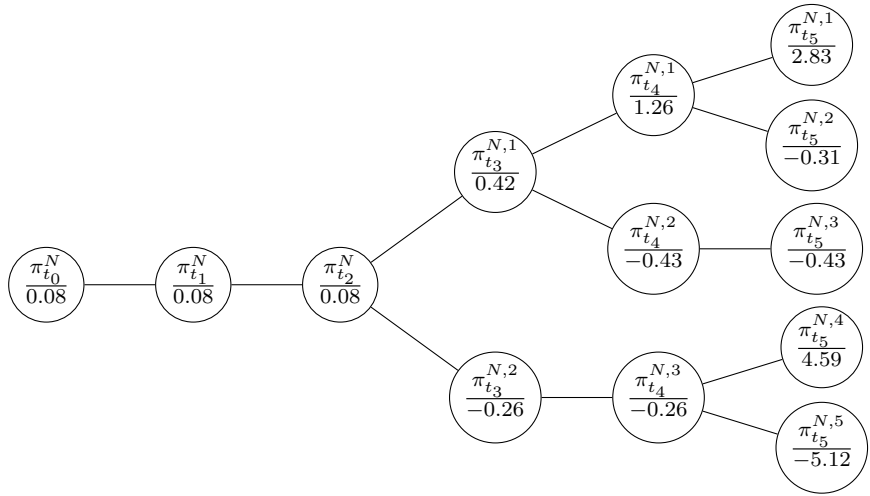
\begin{figure}[H]
	\begin{forest}
    for tree={
   	calign=center,
   	grow'=east, 
   			circle,
   			draw,
   	text height=1.5ex, text depth=0.15ex, 
   	l sep+=1.25em,
   	s sep-=0.2em,
   	inner sep=0.25em
   }
		[$\frac{\pi_{t_0}^{N}}{0.08}$
		  [$\frac{\pi_{t_1}^{N}}{0.08}$
		    [$\frac{\pi_{t_2}^{N}}{0.08}$
		      [$\frac{\pi_{t_3}^{N,1}}{0.42}$
		        [$\frac{\pi_{t_4}^{N,1}}{1.26}$
		           [$\frac{\pi_{t_5}^{N,1}}{2.83}$]
		           [$\frac{\pi_{t_5}^{N,2}}{-0.31}$]]
		        [$\frac{\pi_{t_4}^{N,2}}{-0.43}$
		           [$\frac{\pi_{t_5}^{N,3}}{-0.43}$]
		        ]
		      ]
		      [$\frac{\pi_{t_3}^{N,2}}{-0.26}$
		        [$\frac{\pi_{t_4}^{N,3}}{-0.26}$
		           [$\frac{\pi_{t_5}^{N,4}}{4.59}$]
		           [$\frac{\pi_{t_5}^{N,5}}{-5.12}$]
		        ]
		      ]
		    ]
		  ]
		]
	\end{forest}
	\caption{The decision tree of optimal strategies  associated to the model in \ref{fig1}.}\label{fig1_1}
\end{figure}
\begin{figure}[H]
	\includegraphics[width=0.95\textwidth]{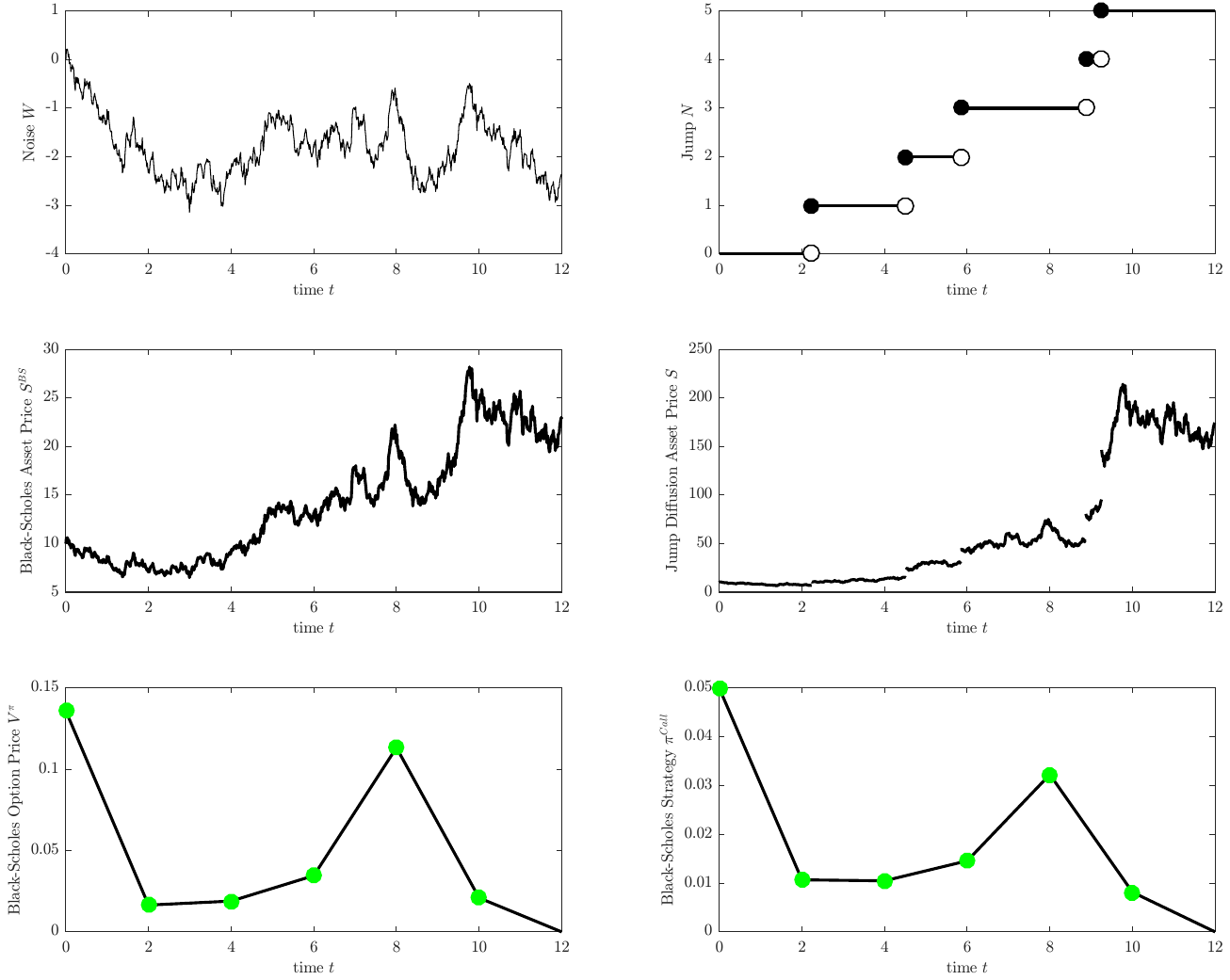}
	\caption{The underlying asset price, European call option price and strategy, the noise and jumps for constants $\mu = 0.15, \sigma = 0.25, \xi\equiv 0.5, T=12, \kappa = 0.1$ and $\lambda = 0.3$.}\label{fig2}
\end{figure}
\begin{figure}[H]
	\begin{forest}
for tree={
	calign=center,
	grow'=east, 
	circle,
	draw,
	text height=1.5ex, text depth=0.15ex, 
	l sep+=1.25em,
	s sep-=0.2em,
	inner sep=0.2em
}
[$\frac{\pi_{t_0}^{N}}{0.05}$
  [$\frac{\pi_{t_1}^{N,1}}{0.23}$
    [$\frac{\pi_{t_2}^{N,1}}{1.67}$
      [$\frac{\pi_{t_3}^{N,1}}{11.31}$
        [$\frac{\pi_{t_4}^{N,1}}{78.00}$
          [$\frac{\pi_{t_5}^{N,1}}{523.34}$]
          [$\frac{\pi_{t_5}^{N,2}}{-367.34}$]
        ]
        [$\frac{\pi_{t_4}^{N,2}}{-55.39}$
          [$\frac{\pi_{t_5}^{N,3}}{-55.39}$]
        ]
      ]
      [$\frac{\pi_{t_3}^{N,2}}{-7.97}$
        [$\frac{\pi_{t_4}^{N,3}}{-7.97}$
          [$\frac{\pi_{t_5}^{N,4}}{-7.97}$]
        ]
      ]
    ]
    [$\frac{\pi_{t_2}^{N,2}}{-1.20}$
      [$\frac{\pi_{t_3}^{N,3}}{-1.20}$
        [$\frac{\pi_{t_4}^{N,4}}{-1.20}$
          [$\frac{\pi_{t_5}^{N,5}}{-1.20}$]
        ]
      ]
    ]
  ]
  [$\frac{\pi_{t_1}^{N,2}}{-0.13}$
    [$\frac{\pi_{t_2}^{N,3}}{-0.13}$
      [$\frac{\pi_{t_3}^{N,4}}{-0.13}$
        [$\frac{\pi_{t_4}^{N,5}}{-0.13}$
          [$\frac{\pi_{t_5}^{N,6}}{-0.13}$]
        ]
      ]
    ]
  ]
]
	\end{forest}
	\caption{The decision tree of optimal strategies  associated to the model in \ref{fig2}.}\label{fig2_1}
\end{figure}
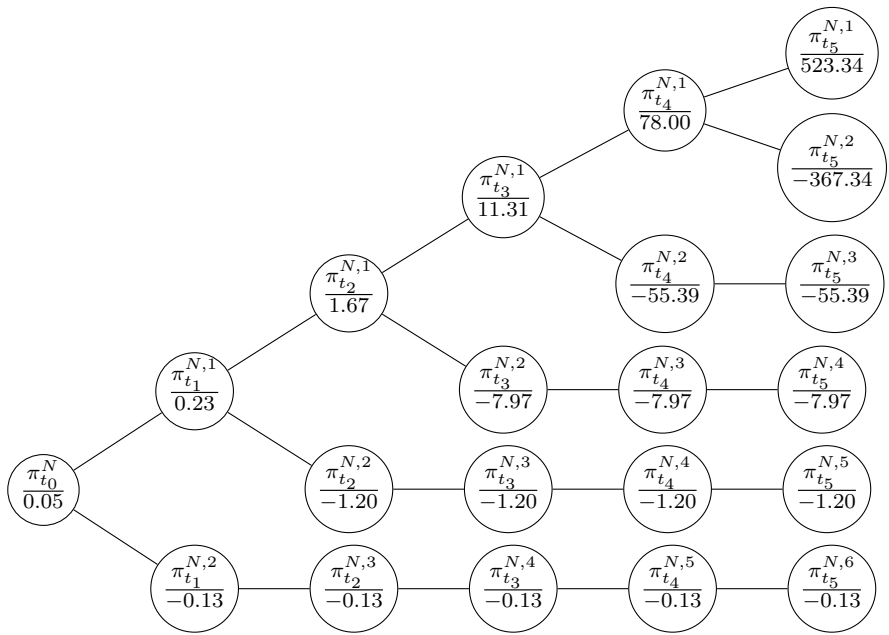
\bibliographystyle{siam}

\begin{thebibliography}{10}

\bibitem{almani2024prediction}
{\sc H.~M. Almani, F.~Shokrollahi, and T.~Sottinen}, {\em Prediction of
  {G}aussian {V}olterra processes with compound {P}oisson jumps}, Statist.
  Probab. Lett., 208 (2024), pp.~Paper No. 110054, 8.

\bibitem{Azmoodeh-2013}
{\sc E.~Azmoodeh}, {\em On the fractional {B}lack--{S}choles market with
  transaction costs}, Communications in Mathematical Finance, 2 (2013),
  pp.~21--40.

\bibitem{black1973pricing}
{\sc F.~Black and M.~Scholes}, {\em The pricing of options and corporate
  liabilities}, Journal of political economy, 81 (1973), pp.~637--654.

\bibitem{carr2002fine}
{\sc P.~Carr, H.~Geman, D.~B. Madan, and M.~Yor}, {\em The fine structure of
  asset returns: An empirical investigation}, The Journal of Business, 75
  (2002), pp.~305--332.

\bibitem{das1996exact}
{\sc S.~R. Das and S.~Foresi}, {\em Exact solutions for bond and option prices
  with systematic jump risk}, Review of derivatives research, 1 (1996),
  pp.~7--24.

\bibitem{Denis-Kabanov-2010}
{\sc E.~Denis and Y.~Kabanov}, {\em Mean square error for the {L}eland-{L}ott
  hedging strategy: convex pay-offs}, Finance Stoch., 14 (2010), pp.~625--667.

\bibitem{follmer1991hedging}
{\sc H.~F{\"o}llmer and M.~Schweizer}, {\em Hedging of contingent claims under
  incomplete information}, Applied stochastic analysis, 5 (1991), pp.~19--31.

\bibitem{geman2002pure}
{\sc H.~Geman}, {\em Pure jump {L}{\'e}vy processes for asset price modelling},
  Journal of Banking \& Finance, 26 (2002), pp.~1297--1316.

\bibitem{geman2001asset}
{\sc H.~Geman, D.~B. Madan, and M.~Yor}, {\em Asset prices are {B}rownian
  motion: only in business time}, in Quantitative Analysis In Financial
  Markets: Collected Papers of the New York University Mathematical Finance
  Seminar (Volume II), World Scientific, 2001, pp.~103--146.

\bibitem{geman2001time}
{\sc H.~Geman, D.~B. Madan, and M.~Yor}, {\em Time changes for {L}{\'e}vy
  processes}, Mathematical Finance, 11 (2001), pp.~79--96.

\bibitem{ingersoll1976theoretical}
{\sc J.~E. Ingersoll~Jr}, {\em A theoretical and empirical investigation of the
  dual purpose funds: An application of contingent-claims analysis}, Journal of
  Financial Economics, 3 (1976), pp.~83--123.

\bibitem{ito1944integral}
{\sc K.~It{\^o}}, {\em {S}tochastic integral}, Proceedings of the Imperial
  Academy, 20 (1944), pp.~519--524.

\bibitem{ito1951sde}
\leavevmode\vrule height 2pt depth -1.6pt width 23pt, {\em On stochastic
  differential equations}, vol.~4, American Mathematical Society New York,
  1951.

\bibitem{ito2006essentials}
\leavevmode\vrule height 2pt depth -1.6pt width 23pt, {\em Essentials of
  stochastic processes}, vol.~231, American Mathematical Soc., 2006.

\bibitem{Kabanov-Safarian-2009}
{\sc Y.~Kabanov and M.~Safarian}, {\em Markets with transaction costs},
  Springer Finance, Springer-Verlag, Berlin, 2009.
\newblock Mathematical theory.

\bibitem{kou2001option}
{\sc S.~Kou and H.~Wang}, {\em Option pricing under a jump-diffusion model},
  (2001).

\bibitem{kou2002jump}
{\sc S.~G. Kou}, {\em A jump-diffusion model for option pricing}, Management
  science, 48 (2002), pp.~1086--1101.

\bibitem{lamberton2011introduction}
{\sc D.~Lamberton and B.~Lapeyre}, {\em Introduction to stochastic calculus
  applied to finance}, Chapman and Hall/CRC, 2011.

\bibitem{leland1985option}
{\sc H.~E. Leland}, {\em Option pricing and replication with transactions
  costs}, The {J}ournal of {F}inance, 40 (1985), pp.~1283--1301.

\bibitem{madan2001financial}
{\sc D.~Madan}, {\em Financial modeling with discontinuous price processes},
  L{\'e}vy processes--theory and applications. Birkhauser, Boston,  (2001).

\bibitem{mandelbrot1997variation}
{\sc B.~B. Mandelbrot}, {\em The variation of certain speculative prices},
  Springer, 1997.

\bibitem{mandelbrot1968fractional}
{\sc B.~B. Mandelbrot and J.~W. Van~Ness}, {\em Fractional {B}rownian motions,
  fractional noises and applications}, SIAM review, 10 (1968), pp.~422--437.

\bibitem{merton1973theory}
{\sc R.~C. Merton}, {\em Theory of rational option pricing}, The Bell Journal
  of economics and management science,  (1973), pp.~141--183.

\bibitem{MERTON1976125}
{\sc R.~C. Merton}, {\em Option pricing when underlying stock returns are
  discontinuous}, Journal of Financial Economics, 3 (1976), pp.~125--144.

\bibitem{ng1969table}
{\sc E.~W. Ng and M.~Geller}, {\em A table of integrals of the error
  functions}, Journal of Research of the National Bureau of Standards B, 73
  (1969), pp.~1--20.

\bibitem{samoradnitsky-Taqqu-1994}
{\sc G.~Samoradnitsky and M.~S. Taqqu}, {\em Stable non-{G}aussian random
  processes: stochastic models with infinite variance}, Chapman \& Hall, CRC,
  1st~ed., 1994.

\bibitem{samuelson1965rational}
{\sc P.~A. Samuelson}, {\em Rational theory of warrant pricing}, in {H}enry
  {P}. {M}c{K}ean {J}r. {S}electa, Springer, 1965, pp.~195--232.

\bibitem{schweizer1992mean}
{\sc M.~Schweizer}, {\em Mean-variance hedging for general claims}, The
  {A}nnals of {A}pplied {P}robability,  (1992), pp.~171--179.

\bibitem{schweizer1994approximating}
\leavevmode\vrule height 2pt depth -1.6pt width 23pt, {\em Approximating random
  variables by stochastic integrals}, The {A}nnals of {P}robability,  (1994),
  pp.~1536--1575.

\bibitem{schweizer1995minimal}
\leavevmode\vrule height 2pt depth -1.6pt width 23pt, {\em On the minimal
  martingale measure and the {F}{\"o}llmer-{S}chweizer decomposition},
  Stochastic analysis and applications, 13 (1995), pp.~573--599.

\bibitem{Shokrollahi-Kilicman-Magdziarz-2016}
{\sc F.~Shokrollahi, A.~K$\mathrm{\i{} l\i{}\c{c}}$man, and M.~Magdziarz}, {\em
  Pricing {E}uropean options and currency options by time changed mixed
  fractional {B}rownian motion with transaction costs}, Int. J. Financ. Eng., 3
  (2016), pp.~1650003, 22.

\bibitem{SHOKROLLAHI201785}
{\sc F.~Shokrollahi and T.~Sottinen}, {\em Hedging in fractional
  {B}lack--{S}choles model with transaction costs}, Statistics \& Probability
  Letters, 130 (2017), pp.~85--91.

\bibitem{Sottinen-Viitasaari-2016-preprintb}
{\sc T.~{Sottinen} and L.~{Viitasaari}}, {\em {Prediction Law of fractional
  Brownian Motion}}, ArXiv e-prints,  (2016).

\bibitem{sottinen2018conditional}
{\sc T.~Sottinen and L.~Viitasaari}, {\em Conditional-mean hedging under
  transaction costs in {G}aussian models}, International journal of theoretical
  and applied finance, 21 (2018), p.~1850015.

\bibitem{Wang-2010a}
{\sc X.-T. Wang}, {\em Scaling and long-range dependence in option pricing {I}:
  pricing {E}uropean option with transaction costs under the fractional
  {B}lack-{S}choles model}, Phys. A, 389 (2010), pp.~438--444.

\bibitem{Wang-2010b}
\leavevmode\vrule height 2pt depth -1.6pt width 23pt, {\em Scaling and long
  range dependence in option pricing, {IV}: pricing {E}uropean options with
  transaction costs under the multifractional {B}lack-{S}choles model}, Phys.
  A., 389 (2010), pp.~789--796.

\bibitem{Wang-Yan-Tang-Zhu-2010}
{\sc X.-T. Wang, H.-G. Yan, M.-M. Tang, and E.-H. Zhu}, {\em Scaling and
  long-range dependence in option pricing {III}: a fractional version of the
  {M}erton model with transaction costs}, Phys. A, 389 (2010), pp.~452--458.

\bibitem{Wang-Zhu-Tang-Yan-2010}
{\sc X.-T. Wang, E.-H. Zhu, M.-M. Tang, and H.-G. Yan}, {\em Scaling and
  long-range dependence in option pricing {II}: pricing {E}uropean option with
  transaction costs under the mixed {B}rownian-fractional {B}rownian model},
  Phys. A, 389 (2010), pp.~445--451.

\end{thebibliography}

\end{document}